\documentclass{amsart}
\usepackage{amsmath, amsthm, amscd, amsfonts, amssymb, graphicx, color}
\input{mathrsfs.sty}
\usepackage{enumitem}
\usepackage[bookmarksnumbered, colorlinks, plainpages]{hyperref}
\newtheorem{theorem}{Theorem}[section]
\newtheorem{lemma}[theorem]{Lemma}
\newtheorem{proposition}[theorem]{Proposition}
\newtheorem{corollary}[theorem]{Corollary}
\theoremstyle{definition}
\newtheorem{definition}[theorem]{Definition}

\theoremstyle{remark}
\newtheorem{remark}[theorem]{Remark}
\numberwithin{equation}{section}
\newcommand{\X}{ \mathbf{X}}
\newcommand{\A}{ \mathbf{A}}
\newcommand{\Y}{ \mathbf{Y}}
\newcommand{\I}{ \mathbf{I}}
\begin{document}

\title[  A Hierarchy of Entanglement Cones]{ A Hierarchy of Entanglement Cones via Rank-Constrained $C^*$-Convex Hulls}

\author{Mohsen Kian}
\address{Department of Mathematics, University of Bojnord, P. O. Box 1339, Bojnord 94531, Iran}
\email{kian@ub.ac.ir }

\subjclass[2020]{Primary 15A69, 46L07; Secondary 81P45}
\keywords{Separable cone,  positive partial transpose cone, $C^*$-convex, operator Schmidt rank}

\begin{abstract}
This paper systematically investigates the geometry of fundamental quantum cones, the separable cone ($\mathscr{P}_+$) and the Positive Partial Transpose (PPT) cone ($\mathcal{P}_{\mathrm{PPT}}$), under generalized non-commutative convexity. We   demonstrate  a sharp stability dichotomy   analyzing $C^*$-convex hulls  of these cones: while $\mathscr{P}_+$ remains stable under local $C^*$-convex combinations, its global $C^*$-convex hull collapses entirely to the cone of all positive semidefinite matrices, $\operatorname{MCL}(\mathscr{P}_+) = \mathscr{P}_0$. To gain finer control and classify intermediate structures, we introduce the concept of ``$k$-$C^*$-convexity'', by using the operator Schmidt rank of $C^*$-coefficients. This constraint defines a new hierarchy of nested intermediate cones, $\operatorname{MCL}_k(\cdot)$. We prove that this hierarchy precisely recovers the known Schmidt number cones for the separable case, establishing a generalized convexity characterization: $\operatorname{MCL}_k(\mathscr{P}_+) = \mathcal{T}_k$. Applied to the PPT cone, this framework generates a family of conjectured non-trivial intermediate cones, $\mathcal{C}_{\mathrm{PPT}, k}$.
\end{abstract}
\maketitle

\section{Introduction}

The geometry of the cone of positive semidefinite matrices, denoted $\mathscr{P}_0$, plays a fundamental role in matrix analysis and quantum information theory \cite{Ando2004, Horodecki1997}. Within this framework, the distinction between separable and entangled states is central to the characterization of quantum correlations. The cone of separable operators, $\mathscr{P}_+$, forms a proper subcone of $\mathscr{P}_0$, structured as the convex hull of product states \cite{Horodecki1996, TerhalHorodecki2000}. The dual of this cone, $\mathscr{P}_-$, consists of block-positive operators (entanglement witnesses), which are essential tools for detecting entanglement \cite{Horodecki1996, Peres1996}.

Despite its fundamental importance, the geometry of $\mathscr{P}_+$ remains notoriously difficult to characterize. The weak membership problem for the separable cone is NP-hard, a fact that has driven recent research into relaxation hierarchies and tensor optimization methods. For instance, recent work by Xu et al. \cite{Xu2025} has highlighted the computational intractability of projecting onto the cone of positive semidefinite tensors, necessitating new convex relaxation techniques. Similarly, geometric approaches using spectral matrix cones have been developed to better approximate these sets in high-dimensional optimization problems \cite{Oliveira2025}.

While the \emph{linear} convexity of these sets is the standard object of study, non-commutative generalizations of convexity offer a richer algebraic perspective that respects the operator structure of quantum mechanics. The notion of \emph{$C^*$-convexity}, introduced by Loebl and Paulsen \cite{LoeblPaulsen}, generalizes the concept of  classical convexity by allowing coefficients to be operators rather than scalars. Specifically, a set is $C^*$-convex if it is closed under conjugations $\sum A_i^* X_i A_i$ where the coefficients satisfy the structural condition $\sum A_i^* A_i = I$. Recent developments by Kennedy, Kim, and Manor \cite{Kennedy2021} and Kumar and Pandey \cite{Kumar2025} have confirmed $C^*$-convexity as the natural framework for analyzing stability properties of operator systems and completely positive maps.

However, in the specific context of bipartite systems $\mathbb{M}_{m} \otimes \mathbb{M}_{n}$, the interplay between tensor product structures and $C^*$-convexity is largely unexplored. While the full cone $\mathscr{P}_0$ is self-dual and $C^*$-convex, the behavior of the separable cone $\mathscr{P}_+$ and the Positive Partial Transpose (PPT) cone under these generalized operations poses significant questions regarding the algebraic fragility of entanglement classes.

In this paper, we provide a systematic study of the $C^*$-convexity of fundamental quantum cones. Our primary contribution is the introduction and investigation of a generalized convexity hierarchy that interpolates between the separable/PPT cones and the full positive cone.
We first establish a sharp boundary for the global stability of $\mathscr{P}_+$. We demonstrate that the \emph{global} $C^*$-convex hull of the separable cone collapses entirely to the full positive cone, $\operatorname{MCL}(\mathscr{P}_+) = \mathscr{P}_0$.
To gain finer control, we introduce the concept of $k$-$C^*$-convexity, parameterized by using the operator Schmidt rank (tensor rank) $k$ of the generating Kraus operators. This new framework generates two families of intermediate cones: $\mathcal{C}_{k} (\mathscr{P}_+) = \operatorname{MCL}_k(\mathscr{P}_+)$ and $\mathcal{C}_{\mathrm{PPT}, k} = \operatorname{MCL}_k(\mathcal{P}_{\mathrm{PPT}})$.

We prove that this hierarchy naturally recovers the Schmidt number cones $\mathcal{T}_k$ studied by Terhal and Horodecki \cite{TerhalHorodecki2000}. Specifically, we show that for the separable cone, $\mathcal{C}_{k} (\mathscr{P}_+) = \mathcal{T}_k$.
For the PPT cone, which contains the bound entangled states, the hierarchy offers a new avenue for classification.

\section{Preliminaries}
 Throughout the paper we work with finite-dimensional complex Hilbert spaces
$\mathbb{C}^{m}$ and $\mathbb{C}^{n}$ and their Hilbert space tensor product
$\mathbb{C}^{m}\otimes\mathbb{C}^{n}$. Let $\mathbb{M}_m$ and $\mathbb{M}_n$ denote the algebras of complex
$m \times m$ and $n \times n$ matrices, respectively, and set
\[
\mathbb{M}_{mn} := \mathbb{M}_m \otimes \mathbb{M}_n.
\]
 While we denote by uppercase letters for matrices in $\mathbb{M}_m$, we will use    bold uppercase letters to  emphasize  elements in  $\mathbb{M}_m\otimes \mathbb{M}_n$.
The full   cone $\mathbb{M}_{mn}^{+}$ of positive semidefinite matrices  is denoted by
$\mathscr{P}_0$.

A positive semidefinite operator $\X\in\mathbb{M}_{mn}^{+}$ is
\emph{separable} if it admits a representation
\[
\X=\sum_{j=1}^{N} \lambda_j\, (a_j\otimes b_j)(a_j\otimes b_j)^{*},
\qquad \lambda_j\ge 0.
\]
The cone of separable positive matrices is denoted by $\mathscr{P}_{+}$.
 It is well-known that $\mathscr{P}_0$ is self-dual under the duality coupling   $\langle X,Y\rangle=\mathrm{Tr}(X^*Y)$. Furthermore, we denote by $\mathscr{P}_-$ the dual cone of $\mathscr{P}_+$. i.e.
\[
\mathscr{P}_-
:= \left\{ \mathbf{W} = \mathbf{W}^\ast \in \mathbb{M}_{mn} :
\operatorname{Tr}(\mathbf{W}\,\mathbf{X}) \ge 0 \ \text{for all }
\mathbf{X} \in \mathscr{P}_+ \right\}.
\]
It is well known    \cite{Ando2004} that $\mathbf{W} \in \mathscr{P}_-$ if and
only if
\[
(z^\ast \otimes w^\ast)\,\mathbf{W}\,(z \otimes w) \ge 0
\quad\text{for all } z \in \mathbb{C}^m,\ w \in \mathbb{C}^n,
\]
i.e.\ $\mathbf{W}$ is \emph{block-positive}.

A set \(\mathcal{K} \subseteq \mathbb{M}_{d}\) is called \emph{$C^*$-convex} (see e.g. \cite{LoeblPaulsen}) if for any finite collection \(\{X_1, \ldots, X_k\} \subseteq \mathcal{K}\) and \(\{A_1, \ldots, A_k\} \subseteq \mathbb{M}_{d}\) satisfying \(\sum_{i=1}^k A_i^* A_i = I\), the sum \(\sum_{i=1}^k A_i^* X_i A_i \in \mathcal{K}\). This notion is a non-commutative extension of the linear convexity. Every $C^*$-convex set is linearly convex and the converse is not true in general. Typical examples of $C^*$-convex sets are the unit ball of the spectral norm and the numerical radius norm, see also \cite{kian2017}. For any set $\mathcal{K}\subseteq\mathbb{M}_d$, adopting the notation of \cite{LoeblPaulsen}, we denote by $\mathrm{MCL}(\mathcal{K})$ the smallest norm-closed $C^*$-convex set containing $\mathcal{K}$.

For a subset $\mathcal{K} \subseteq \mathbb{M}_{d}$, we define its
\emph{conic hull} by
\[
\operatorname{cone}(\mathcal{K})
:= \left\{
\sum_{j=1}^r \lambda_j X_j :
r \in \mathbb{N},\ \lambda_j \ge 0,\
X_j \in \mathcal{K}
\right\}.
\]
The norm-closure of the conic hull of $\mathcal{K}$ will be denoted by $\operatorname{ccone}(\mathcal{K})$.

 \section{ Stability and Collapse of the Separable Cone}
 This section initiates our investigation into the algebraic stability of the separable cone ($\mathscr{P}_+$) under $C^*$-convexity. We first demonstrate that $\mathscr{P}_+$ is robust under local $C^*$-convex combinations, consistent with its interpretation as states invariant under local quantum channels. We then establish the primary result of this section: the complete and surprising collapse of the separable cone to the full positive cone $\mathscr{P}_0$ under global, unconstrained $C^*$-convex combinations.
The following simple observation will be used repeatedly.

\begin{lemma}
\label{lem:MCL-in-PSD}
We have $\operatorname{MCL}(\mathscr{P}_+) \subseteq \mathscr{P}_0$.
\end{lemma}

\begin{proof}
The set $\mathscr{P}_0$ is closed, and it is $C^\ast$-convex: if
$\mathbf{X}_1,\dots,\mathbf{X}_r \in \mathscr{P}_0$ and
$\mathbf{A}_1,\dots,\mathbf{A}_r \in \mathbb{M}_{mn}$ satisfy
$\sum_{i=1}^r \mathbf{A}_i^\ast \mathbf{A}_i = \mathbf{I}$, then each
$\mathbf{A}_i^\ast \mathbf{X}_i \mathbf{A}_i$ is positive semidefinite and
hence their sum belongs to $\mathscr{P}_0$. Clearly
$\mathscr{P}_+ \subseteq \mathscr{P}_0$, so $\mathscr{P}_0$ is a closed
$C^\ast$-convex set containing $\mathscr{P}_+$. By minimality of
$\operatorname{MCL}(\mathscr{P}_+)$, we obtain
$\operatorname{MCL}(\mathscr{P}_+) \subseteq \mathscr{P}_0$.
\end{proof}
The next theorem establishes the primary algebraic fragility of the separable cone, demonstrating that its $C^*$-convex hull strictly enlarges beyond the separable states, ultimately collapsing to the entire positive cone $\mathscr{P}_0$.
\begin{theorem}[Strict enlargement of the separable cone]
\label{thm:strict-MCL}
For all $m,n \ge 2$,
\[
\mathscr{P}_+ \subsetneq \operatorname{MCL}(\mathscr{P}_+).
\]
\end{theorem}

\begin{proof}
The inclusion $\mathscr{P}_+ \subseteq \operatorname{MCL}(\mathscr{P}_+)$
follows directly from the definition of $\operatorname{MCL}(\mathscr{P}_+)$
as the smallest closed $C^\ast$-convex set containing $\mathscr{P}_+$.
It remains to show that the inclusion is strict.

Choose nonzero vectors $u \in \mathbb{C}^m$ and $v \in \mathbb{C}^n$ and
normalize them so that $\|u\| = \|v\| = 1$. Define
\[
\mathbf{X}_0 := (u u^\ast) \otimes (v v^\ast) \in \mathbb{M}_{mn}.
\]
Then $\mathbf{X}_0$ is positive semidefinite and of product form, hence
$\mathbf{X}_0 \in \mathscr{P}_+ \subseteq \operatorname{MCL}(\mathscr{P}_+)$.

Since $m,n \ge 2$, the tensor product space
$\mathbb{C}^m \otimes \mathbb{C}^n$ contains vectors that are not simple
tensors. Choose a unit vector $w \in \mathbb{C}^{mn}$ that is not of the form
$x \otimes y$ with $x \in \mathbb{C}^m$ and $y \in \mathbb{C}^n$. The unitary
group on $\mathbb{C}^{mn}$ acts transitively on the unit sphere, so there
exists a unitary matrix $\mathbf{U} \in \mathbb{M}_{mn}$ such that
\[
\mathbf{U}(u \otimes v) = w.
\]

Define $\mathbf{A}_1 := \mathbf{U}$. Then
$\mathbf{A}_1^\ast \mathbf{A}_1 = \mathbf{U}^\ast \mathbf{U} = \mathbf{I}$,
so the singleton family $\{\mathbf{A}_1\}$ is  a $C^\ast$-convex family.
Consider
\[
\mathbf{W} := \mathbf{A}_1^\ast \mathbf{X}_0 \mathbf{A}_1
            = \mathbf{U}^\ast \mathbf{X}_0 \mathbf{U}
            \in \operatorname{MCL}(\mathscr{P}_+),
\]
where the last inclusion holds by the $C^\ast$-convexity of
$\operatorname{MCL}(\mathscr{P}_+)$ and the fact that
$\mathbf{X}_0 \in \operatorname{MCL}(\mathscr{P}_+)$.

Using $\mathbf{X}_0 = (u u^\ast)\otimes (v v^\ast)$ and
$\mathbf{U}(u\otimes v)=w$, we obtain
\[
\mathbf{W} = \mathbf{U}^\ast \bigl((u u^\ast)\otimes (v v^\ast)\bigr)\mathbf{U}
           = w w^\ast.
\]
Thus $\mathbf{W}$ is a rank-one positive semidefinite matrix.

A rank-one matrix $\mathbf{Z} = z z^\ast \in \mathbb{M}_{mn}$ belongs to
$\mathscr{P}_+$ if and only if $z$ is a simple tensor: indeed, if
$z = x \otimes y$ for some $x \in \mathbb{C}^m$ and $y \in \mathbb{C}^n$ then
$\mathbf{Z} = (x x^\ast)\otimes (y y^\ast)$ is separable, and conversely if
$\mathbf{Z} \in \mathscr{P}_+$ then its range is one-dimensional and spanned
by a product vector. By construction $w$ is not a simple tensor, hence
$\mathbf{W} = w w^\ast \notin \mathscr{P}_+$. We have therefore exhibited
an element $\mathbf{W} \in \operatorname{MCL}(\mathscr{P}_+)$ that is not
in $\mathscr{P}_+$, which proves that the inclusion is strict.
\end{proof}

The following result was  showed in the proof of Theorem~\ref{thm:strict-MCL}.
\begin{corollary}
\label{thm:rank-one-in-MCL}
Let $m,n \ge 2$, and let $\mathbf{P}_w := w w^\ast \in \mathbb{M}_{mn}$ denote
the rank-one projector associated to a unit vector $w \in \mathbb{C}^{mn}$.
Then
\[
\mathbf{P}_w \in \operatorname{MCL}(\mathscr{P}_+)
\quad\text{for every unit } w \in \mathbb{C}^{mn}.
\]
\end{corollary}

 %%%%%%%%%%%%=================================================================

\begin{corollary}
\label{thm:cone-MCL-equals-P0}
Let $m,n \ge 2$. Then
\[
\operatorname{cone}\bigl(\operatorname{MCL}(\mathscr{P}_+)\bigr)
= \mathscr{P}_0.
\]
\end{corollary}

\begin{proof}
We first show
$\operatorname{cone}(\operatorname{MCL}(\mathscr{P}_+)) \subseteq \mathscr{P}_0$.
By Lemma~\ref{lem:MCL-in-PSD} we have
$\operatorname{MCL}(\mathscr{P}_+) \subseteq \mathscr{P}_0$, and
$\mathscr{P}_0$ is a cone. Hence any finite nonnegative linear combination
of elements of $\operatorname{MCL}(\mathscr{P}_+)$ still belongs to
$\mathscr{P}_0$, which proves the inclusion
$\operatorname{cone}(\operatorname{MCL}(\mathscr{P}_+)) \subseteq \mathscr{P}_0$.

For the reverse inclusion, let $\mathbf{Y} \in \mathscr{P}_0$ be arbitrary.
By the spectral theorem there exist nonnegative scalars
$\lambda_1,\dots,\lambda_r$ and an orthonormal family of vectors
$w_1,\dots,w_r \in \mathbb{C}^{mn}$ such that
\[
\mathbf{Y} = \sum_{j=1}^r \lambda_j \mathbf{P}_{w_j},
\qquad
\mathbf{P}_{w_j} := w_j w_j^\ast.
\]
By Theorem~\ref{thm:rank-one-in-MCL}, each $\mathbf{P}_{w_j}$ belongs to
$\operatorname{MCL}(\mathscr{P}_+)$, and therefore
$\lambda_j \mathbf{P}_{w_j} \in
\operatorname{cone}(\operatorname{MCL}(\mathscr{P}_+))$ for all $j$.
Since the conic hull is closed under finite sums, we have
\[
\mathbf{Y}
= \sum_{j=1}^r \lambda_j \mathbf{P}_{w_j}
\in \operatorname{cone}(\operatorname{MCL}(\mathscr{P}_+)).
\]
Because $\mathbf{Y} \in \mathscr{P}_0$ was arbitrary, this proves
$\mathscr{P}_0 \subseteq \operatorname{cone}(\operatorname{MCL}(\mathscr{P}_+))$,
and hence the desired equality.
\end{proof}
\begin{remark}
The preceding results show that the closed $C^\ast$-convex hull
$\operatorname{MCL}(\mathscr{P}_+)$ of the separable cone is a large subset
of the positive cone: it contains all separable positive semidefinite matrices,
and in fact all rank-one projectors on $\mathbb{C}^m \otimes \mathbb{C}^n$.
Consequently, its conic hull coincides with $\mathscr{P}_0$.

On the other hand, $\operatorname{MCL}(\mathscr{P}_+)$ is defined as a closed
$C^\ast$-convex \emph{set} rather than a cone, so its geometry is quite
different from that of the usual cones considered in entanglement theory.
\end{remark}

%%%%%%%%%%%%%%%%%%%%%%%%%%%%%%%%%%%%%%%%%%%%%%%%%%%%%%%%%%%%%%%%%%%%%%%%%%%%%%

\begin{theorem}
\label{thm:Pminus-not-Cstar-convex}
Assume $m,n \ge 2$.   Then the cone
$\mathscr{P}_-$ is not $C^\ast$-convex.
\end{theorem}

\begin{proof}
Since $\mathscr{P}_+$ is a proper subset of $\mathscr{P}_0$ for $m,n \ge 2$,
we can choose $\mathbf{Y} \in \mathscr{P}_0 \setminus \mathscr{P}_+$. By the
separating hyperplane theorem (We view the Hermitian matrices as a real vector space with inner product $\langle \mathbf{X},\mathbf{Y}\rangle=\mathrm{Tr}(\mathbf{XY})$) applied to the closed convex cone
$\mathscr{P}_+$ and the point $\mathbf{Y} \notin \mathscr{P}_+$, there exists
a Hermitian matrix $\mathbf{W} \in \mathbb{M}_{mn}$ such that
\[
\operatorname{Tr}(\mathbf{W}\,\mathbf{X}) \ge 0
\quad\text{for all } \mathbf{X} \in \mathscr{P}_+,
\qquad
\operatorname{Tr}(\mathbf{W}\,\mathbf{Y}) < 0.
\]
By definition this means $\mathbf{W} \in \mathscr{P}_-$, and $\mathbf{W}$ is
an entanglement witness detecting the entangled positive semidefinite matrix
$\mathbf{Y}$.

Since $\mathbf{Y} \in \mathscr{P}_0$ is positive semidefinite, it admits a
spectral decomposition
\[
\mathbf{Y} = \sum_{i=1}^r \lambda_i\, z_i z_i^\ast,
\qquad
\lambda_i > 0,
\quad
r = \operatorname{rank}(\mathbf{Y}).
\]
Because $\operatorname{Tr}(\mathbf{W}\mathbf{Y}) < 0$, at least one term in
the above sum must satisfy
\[
\langle z_j,\, \mathbf{W} z_j \rangle < 0.
\]
Fix such an index $j$ and set $\mathbf{z} := z_j$.
 We can choose a unit vector $u \otimes v \in \mathbb{C}^m \otimes \mathbb{C}^n$
and a unitary matrix $\mathbf{U} \in \mathbb{M}_{mn}$ such that
\[
\mathbf{U}(u \otimes v) = \frac{\mathbf{z}}{\|\mathbf{z}\|}.
\]

Define $\mathbf{W}' := \mathbf{U}^\ast \mathbf{W} \mathbf{U}$. We will show
that $\mathbf{W}'$ is not block-positive, and hence $\mathbf{W}' \notin
\mathscr{P}_-$. Consider the product vector $u \otimes v$:
\[
(u \otimes v)^\ast \mathbf{W}' (u \otimes v)
= (u \otimes v)^\ast \mathbf{U}^\ast \mathbf{W} \mathbf{U} (u \otimes v)
= \bigl( \mathbf{U}(u \otimes v) \bigr)^\ast
    \mathbf{W}\,
   \bigl( \mathbf{U}(u \otimes v) \bigr).
\]
By construction $\mathbf{U}(u \otimes v) = \mathbf{z} / \|\mathbf{z}\|$, so
\[
(u \otimes v)^\ast \mathbf{W}' (u \otimes v)
= \frac{1}{\|\mathbf{z}\|^2}\,\mathbf{z}^\ast \mathbf{W} \mathbf{z}
< 0.
\]
Thus $\mathbf{W}'$ takes a negative expectation value on the product vector
$u \otimes v$, so $\mathbf{W}'$ is not block-positive and therefore
$\mathbf{W}' \notin \mathscr{P}_-$.

Now observe that $\mathbf{W}'$ is obtained from $\mathbf{W}$ by a one-term
$C^\ast$-combination: if we define $\mathbf{A}_1 := \mathbf{U}$, then
$\mathbf{A}_1^\ast \mathbf{A}_1 = \mathbf{I}$ and
\[
\mathbf{W}' = \mathbf{A}_1^\ast \mathbf{W} \mathbf{A}_1.
\]
We have exhibited $\mathbf{W} \in \mathscr{P}_-$ and a matrix
$\mathbf{A}_1 \in \mathbb{M}_{mn}$ with $\mathbf{A}_1^\ast \mathbf{A}_1
= \mathbf{I}$ such that $\mathbf{A}_1^\ast \mathbf{W} \mathbf{A}_1 \notin
\mathscr{P}_-$. This shows that $\mathscr{P}_-$ is not closed under
$C^\ast$-convex combinations  and hence
$\mathscr{P}_-$ is not $C^\ast$-convex.
\end{proof}

%%%%%%%%%%%%%%%%%%%%%%%%%%%%%%%%%%%%%%%%%%%%%%%%%%%%%%%%%%%%%%%%%%%%%%%
Although $\mathscr{P}_{-}$ is  not globally $C^*$-convex, as shown by the preceding theorem, it enjoys a weaker type of $C^*$-convexity.

Here,  we introduce a variant of matrix convexity that respects the
bipartite tensor-product structure $\mathbb{M}_{mn} = \mathbb{M}_m \otimes
\mathbb{M}_n$. Unlike the global notion of $C^\ast$-convexity considered in the
preceding sections, the definition below restricts Kraus operators to local
tensors $B_i \otimes C_i$.

\begin{definition}
We say that a family of operators
\[
\mathbf{A}_1,\dots,\mathbf{A}_r \in \mathbb{M}_{mn}
\]
is   a \emph{local $C^*$-convex family} if each operator has the form
\[
\mathbf{A}_i = B_i \otimes C_i
\quad\text{with}\quad
B_i \in \mathbb{M}_m,\ C_i \in \mathbb{M}_n,
\]
and the normalization condition
\begin{equation}
\label{eq:local-normalization}
\sum_{i=1}^r \mathbf{A}_i^\ast \mathbf{A}_i
= \sum_{i=1}^r (B_i^\ast B_i)\otimes (C_i^\ast C_i)
= I_m \otimes I_n
\end{equation}
is satisfied.
\end{definition}

\begin{definition}[Local $C^\ast$-convex set]
We say that a subset $\mathcal{K} \subseteq \mathbb{M}_{mn}$ is
\emph{locally $C^\ast$-convex} if for every finite family
$\mathbf{X}_1,\dots,\mathbf{X}_r \in \mathcal{K}$ and every local $C^*$-convex family
$\{\mathbf{A}_i\}_{i=1}^r$ satisfying \eqref{eq:local-normalization}, one has
\[
\sum_{i=1}^r \mathbf{A}_i^\ast \mathbf{X}_i \mathbf{A}_i \in \mathcal{K}.
\]
\end{definition}

\begin{definition}[Local matrix-convex hull]
\label{def:MCL-local}
Let $\mathcal{K} \subseteq \mathbb{M}_{mn}$ .
We define the  \emph{local $C^*$-convex hull} of $\mathcal{K}$, denoted by   $\operatorname{MCL}_{\mathrm{loc}}(\mathcal{K})$ to be  the smallest
closed locally $C^\ast$-convex cone containing $\mathcal{K}$.
\end{definition}

We record a simple but important observation: restricting Kraus operators to
local tensors $\mathbf{B}_i \otimes \mathbf{C}_i$ produces a significantly
weaker notion of $C^\ast$-convexity than allowing arbitrary operators in
$\mathbb{M}_{mn}$. The following example illustrates that the two notions do
not coincide.

\begin{proposition}
\label{prop:block-positive-local-not-global}
Let $\mathscr{P}_- $  denote the cone of block-positive
operators in $\mathbb{M}_{mn}$. Then
\begin{enumerate}
    \item[(i)] $\mathscr{P}_-$ is a locally $C^\ast$-convex cone.
    \item[(ii)] $\mathscr{P}_-$ is \emph{not} globally $C^\ast$-convex.
\end{enumerate}
\end{proposition}

\begin{proof}
\emph{(i) Local $C^\ast$-convexity.}
Suppose $\mathbf{W}_1,\dots,\mathbf{W}_r \in \mathscr{P}_-$ and
$\{\mathbf{A}_i\}_{i=1}^r$ is a local Kraus family,
$\mathbf{A}_i = B_i \otimes C_i$, satisfying
$\sum_i \mathbf{A}_i^\ast \mathbf{A}_i = I_m \otimes I_n$.
Let
\[
\mathbf{W} := \sum_{i=1}^r \mathbf{A}_i^\ast \mathbf{W}_i \mathbf{A}_i.
\]
Let $z \in \mathbb{C}^m$ and $w \in \mathbb{C}^n$. Then
\[
(z^\ast \otimes w^\ast)\,\mathbf{W}\,(z\otimes w)
= \sum_{i=1}^r
\bigl( (B_i z)^\ast \otimes (C_i w)^\ast \bigr)
\,\mathbf{W}_i\,
\bigl( (B_i z) \otimes (C_i w) \bigr)
\ge 0,
\]
because each $\mathbf{W}_i$ is block-positive and each
$(B_i z)\otimes (C_i w)$ is a product vector.
Thus $\mathbf{W} \in \mathscr{P}_-$, proving that $\mathscr{P}_-$ is locally
$C^\ast$-convex.

\smallskip
\emph{(ii)}  Proved in Theorem~\ref{thm:Pminus-not-Cstar-convex}.
\end{proof}

Before proving the dramatic global collapse, we establish a necessary stability result. The following proposition shows that $\mathscr{P}_+$ is preserved under $C^*$-convex combinations when the coefficients are restricted to local operators, confirming its expected stability under local quantum operations.
\begin{proposition}
\label{prop:Pplus-local-Cstar-convex}
The separable cone $\mathscr{P}_+$ is a closed locally $C^\ast$-convex cone.
\end{proposition}

\begin{proof}
It is clear from the definition that $\mathscr{P}_+$ is a convex cone and that
it is closed in the norm topology of $\mathbb{M}_{mn}$.

We verify local $C^\ast$-convexity. Let
$\mathbf{X}_1,\dots,\mathbf{X}_r \in \mathscr{P}_+$ and let
$\{\mathbf{A}_i\}_{i=1}^r$ be a local Kraus family,
$\mathbf{A}_i = B_i \otimes C_i$, satisfying
\[
\sum_{i=1}^r \mathbf{A}_i^\ast \mathbf{A}_i
= \sum_{i=1}^r (B_i^\ast B_i)\otimes
                 (C_i^\ast C_i)
= I_m \otimes I_n.
\]

For each $i$ we may write
\[
\mathbf{X}_i = \sum_{j=1}^{r_i} P_{ij} \otimes Q_{ij},
\quad
P_{ij} \ge 0,\ Q_{ij} \ge 0.
\]
Consider
\[
\mathbf{Y}
:= \sum_{i=1}^r \mathbf{A}_i^\ast \mathbf{X}_i \mathbf{A}_i
 = \sum_{i=1}^r (B_i^\ast \otimes C_i^\ast)
      \left( \sum_{j=1}^{r_i} P_{ij} \otimes Q_{ij} \right)
      (B_i \otimes C_i).
\]
Using $(B_i^\ast \otimes C_i^\ast)
(P_{ij} \otimes Q_{ij})
(B_i \otimes C_i)
= (B_i^\ast P_{ij} B_i)
  \otimes
  (C_i^\ast Q_{ij} C_i)$, we obtain
\[
\mathbf{Y}
= \sum_{i=1}^r \sum_{j=1}^{r_i}
   \bigl(B_i^\ast P_{ij} B_i\bigr)
   \otimes
   \bigl(C_i^\ast Q_{ij} C_i\bigr).
\]
Each matrix $B_i^\ast P_{ij} B_i$ is positive
semidefinite on $\mathbb{M}_m$, and each
$C_i^\ast Q_{ij} \mathbf{C}_i$ is positive semidefinite on
$\mathbb{M}_n$. Thus $\mathbf{Y}$ is a finite sum of tensor products of
positive semidefinite matrices, and hence $\mathbf{Y} \in \mathscr{P}_+$.
This proves that $\mathscr{P}_+$ is locally $C^\ast$-convex.
\end{proof}

Building upon the previous result, the following theorem confirms the exact local stability of the separable cone, showing that its local $C^*$-convex hull is precisely the cone itself.
\begin{theorem}
\label{thm:MCLloc-equals-Pplus}
Let $m,n \ge 2$. Then the local matrix-convex hull of the separable cone
coincides with the separable cone:
\[
\operatorname{MCL}_{\mathrm{loc}}(\mathscr{P}_+)
= \mathscr{P}_+.
\]
\end{theorem}

\begin{proof}
By Definition~\ref{def:MCL-local},
$\operatorname{MCL}_{\mathrm{loc}}(\mathscr{P}_+)$ is the intersection of all
closed locally $C^\ast$-convex cones $\mathcal{K} \subseteq \mathbb{M}_{mn}$
such that $\mathscr{P}_+ \subseteq \mathcal{K}$. In particular, by
Proposition~\ref{prop:Pplus-local-Cstar-convex}, the separable cone
$\mathscr{P}_+$ itself is a closed locally $C^\ast$-convex cone containing
$\mathscr{P}_+$. Therefore
\[
\operatorname{MCL}_{\mathrm{loc}}(\mathscr{P}_+)
\subseteq \mathscr{P}_+.
\]

On the other hand, by construction we always have
$\mathscr{P}_+ \subseteq \operatorname{MCL}_{\mathrm{loc}}(\mathscr{P}_+)$,
since $\operatorname{MCL}_{\mathrm{loc}}(\mathscr{P}_+)$ is defined as the
smallest closed locally $C^\ast$-convex cone containing $\mathscr{P}_+$. Thus
\[
\mathscr{P}_+ \subseteq \operatorname{MCL}_{\mathrm{loc}}(\mathscr{P}_+)
\subseteq \mathscr{P}_+,
\]
which implies the desired equality
$\operatorname{MCL}_{\mathrm{loc}}(\mathscr{P}_+) = \mathscr{P}_+$.
\end{proof}

We summarize the preceding analysis in the following no-go statement.

\begin{theorem}\label{thm:no-go}
\begin{enumerate}
\item[(a)]
Let $\mathcal{C}_{\mathrm{glob}}$ be the smallest closed cone
$\mathcal{K} \subseteq \mathbb{M}_{mn}$ such that
\begin{enumerate}[label=(\roman*)]
    \item $\mathscr{P}_+ \subseteq \mathcal{K}$,
    \item $\mathcal{K}$ is $C^\ast$-convex (with respect to arbitrary $C^*$-convex   families $\{\mathbf{A}_i\} \subset \mathbb{M}_{mn}$).
\end{enumerate}
Then
\[
\mathcal{C}_{\mathrm{glob}} = \mathscr{P}_0.
\]

\item[(b)]
Let $\mathcal{C}_{\mathrm{loc}}$ be the smallest closed cone
$\mathcal{K} \subseteq \mathbb{M}_{mn}$ such that
\begin{enumerate}[label=(\roman*)]
    \item $\mathscr{P}_+ \subseteq \mathcal{K}$,
    \item $\mathcal{K}$ is locally $C^\ast$-convex (with respect to local $C^*$-convex
    families of the form $\mathbf{A}_i = B_i \otimes C_i$).
\end{enumerate}
Then
\[
\mathcal{C}_{\mathrm{loc}} = \mathscr{P}_+.
\]
\end{enumerate}

In particular, under these two natural $C^\ast$-convexity schemes there is no
nontrivial closed convex cone $\mathcal{K}$ satisfying
\[
\mathscr{P}_+ \subsetneq \mathcal{K} \subsetneq \mathscr{P}_0.
\]
\end{theorem}
\begin{proof}
(a) is exactly Corollary~\ref{thm:cone-MCL-equals-P0}.
(b) is exactly Theorem~\ref{thm:MCLloc-equals-Pplus}.
\end{proof}

\section{The $k$-$C^*$-Convexity Hierarchy and Schmidt Number Cones}
The dramatic global collapse of the separable cone necessitates a mechanism for control. This section introduces the concept of $k$-$C^*$-convexity by constraining the $C^*$-coefficients via their operator Schmidt rank. The core result here establishes that this new algebraic hierarchy precisely characterizes the known Schmidt number cones ($\mathcal{T}_k$), thus bridging generalized non-commutative convexity with the resource theory of entanglement.

First, we recall some notions.
For a vector $v \in \mathbb{C}^m \otimes \mathbb{C}^n$, we denote by
$
\operatorname{SR}(v)
$
the smallest integer $r \ge 1$ for which there exist vectors
$x_1,\dots,x_r \in \mathbb{C}^m$ and $y_1,\dots,y_r \in \mathbb{C}^n$
such that
\[
v = \sum_{i=1}^r x_i \otimes y_i.
\]
For an operator $\A \in \mathbb{M}_{mn} \cong \mathbb{M}_m \otimes \mathbb{M}_n$, we denote by
$
\operatorname{OSR}(\A)
$
the \emph{operator Schmidt rank} of~$\A$, i.e., the smallest integer
$r \ge 1$ such that
\[
\A = \sum_{i=1}^r R_i \otimes S_i
\]
for some matrices $R_i \in \mathbb{M}_m$ and $S_i \in \mathbb{M}_n$.

\begin{lemma}\label{srank}
Let   $\operatorname{OSR}(\A) \le k$  and let
$v \in \mathbb{C}^m \otimes \mathbb{C}^n$ satisfy
$\operatorname{SR}(v) = r$. Then
\[
\operatorname{SR}(\A v) \;\le\;\operatorname{OSR}(\A) \operatorname{SR}(v).
\]
In particular, if $v$ is a simple tensor
(i.e.\ $\operatorname{SR}(v)=1$), then
$\operatorname{SR}(\A v) \le k$.
\end{lemma}

\begin{proof}
Since $\operatorname{OSR}(\A) \le k$, there exist matrices
$R_1,\dots,R_k \in \mathbb{M}_m$ and $S_1,\dots,S_k \in \mathbb{M}_n$ such that
\[
\A = \sum_{i=1}^k R_i \otimes S_i.
\]
Since $\operatorname{SR}(v) = r$, there exist vectors
$x_1,\dots,x_r \in \mathbb{C}^m$ and $y_1,\dots,y_r \in \mathbb{C}^n$ such that
\[
v = \sum_{j=1}^r x_j \otimes y_j.
\]
Then
\[
\A v
= \sum_{i=1}^k \sum_{j=1}^r
   (R_i x_j) \otimes (S_i y_j),
\]
so $\A v$ is written as a sum of at most $k r$ simple tensors.
Hence $\operatorname{SR}(\A v) \le k r$, as claimed.
The special case $\operatorname{SR}(v)=1$ is immediate.
\end{proof}

\begin{definition}[Tensor Rank $k$ Matrix Class $\mathcal{A}_k$]
For a positive integer $k$, define the tensor rank $k$ matrix class $\mathcal{A}_k$ by
\[
\mathcal{A}_k
    := \Bigl\{
        \X \in \mathbb{M}_{mn} \,\Big|\,
        \mathrm{OSR}(\X)\leq k
      \Bigr\}
\]
and the \emph{Schmidt number $k$} cone  by
\[
\mathcal{T}_k
:= \operatorname{cone}\Bigl\{
  v v^{*} \in \mathbb{M}_{mn} :
  v \in \mathbb{C}^{m}\otimes\mathbb{C}^{n},\
  \operatorname{SR}(v)\le k
\Bigr\}.
\]
\end{definition}

The cone $\mathcal{T}_k$ consist  of positive semidefinite matrices,  whose
\emph{Schmidt number} is at most $k$, in the sense of Terhal and
Horodecki.  It satisfies
\[
\mathcal{T}_1 = \mathscr{P}_{+},
\qquad
\mathcal{T}_{d} = \mathscr{P}_{0},
\qquad
\mathcal{T}_{1} \subsetneq \mathcal{T}_{2} \subsetneq \cdots \subsetneq \mathcal{T}_{d},
\]
where $d=\min\{m,n\}$.
The dual cone $\mathcal{T}_k^{*}$ consists of all block-positive
operators with respect to vectors of Schmidt rank at most $k$, and is
well known to correspond (via the Jamiołkowski isomorphism) to the cone
of $k$-positive linear maps.

\begin{definition}[$k$-$C^*$-Convex Hulls]
We define the {\emph{$k$-$C^*$-convex hull}} of a cone $\mathcal{K} \subset \mathbb{M}_{mn}$, denoted by $\operatorname{MCL}_k(\mathcal{K})$, to be the smallest closed cone containing $\mathcal{K}$,  which is invariant under all $C^*$-convex combinations whose generating matrices $\mathbf{A}_i$ are coming from  the class $\mathcal{A}_k$:
\[
\operatorname{MCL}_k(\mathcal{K}) = \operatorname{ccone} \left\{ \sum_{i=1}^N \mathbf{A}_i^* \mathbf{X}_i \mathbf{A}_i : N\in\mathbb{N},\,\,\mathbf{X}_i \in \mathcal{K},\, \mathbf{A}_i \in \mathcal{A}_k\,\text{ and } \, \sum_{i=1}^N \mathbf{A}_i^* \mathbf{A}_i = \mathbf{I} \right\}.
\]
Moreover, we define the \emph{conic $k$--$C^{*}$-convex hull} of  $\mathcal{K}$ by
\[
\operatorname{CMCL}_k(\mathcal{K})
:=  \operatorname{ccone}\left\{ \sum_{i=1}^N \mathbf{A}_i^* \mathbf{X}_i \mathbf{A}_i : N\in\mathbb{N},\,\,\mathbf{X}_i \in \mathcal{K},\, \mathbf{A}_i \in \mathcal{A}_k\,\text{ and } \, \sum_{i=1}^N \mathbf{A}_i^* \mathbf{A}_i \leq \mathbf{I} \right\}.
\]
\end{definition}

The next result  establishes the central equivalence of the section, showing that the hierarchy generated by $k$-$C^*$-convexity precisely coincides with the known Schmidt number cones $\mathcal{T}_k$, thereby providing a powerful algebraic characterization for these entanglement classes.
\begin{theorem}[Normalized $k$--$C^{*}$-hull and Schmidt number]
\label{thm:normalized-MCLk-inclusion}
For each $k\in\{1,\dots,d\}$ we have
\begin{align}\label{mclk-subseteq-Tk}
  \operatorname{MCL}_k(\mathscr{P}_{+})
  \subseteq \mathcal{T}_k
\end{align}
and
\begin{align}\label{cmcl=Tk}
 \operatorname{CMCL}_k(\mathscr{P}_{+})
  = \mathcal{T}_k.
\end{align}
Moreover,
\[
  \operatorname{CMCL}_1(\mathscr{P}_{+})=\operatorname{MCL}_1(\mathscr{P}_{+}) = \mathscr{P}_{+} = \mathcal{T}_1,
  \qquad
  \operatorname{CMCL}_d(\mathscr{P}_{+})=\operatorname{MCL}_d(\mathscr{P}_{+}) = \mathscr{P}_{0} = \mathcal{T}_d.
\]
\end{theorem}
\begin{proof}
We first prove \eqref{mclk-subseteq-Tk}.
Let $\Y \in \operatorname{MCL}_k(\mathscr{P}_{+})$.
By definition there exist $\X_i \in \mathscr{P}_{+}$ and
$\A_i \in \mathcal{A}_k$ such that
\[
  \Y
  = \sum_{i=1}^{N} \A_i^{*} \X_i \A_i
  \qquad\mbox{and}\qquad
  \sum_{i=1}^{N} \A_i^{*} \A_i = I,
\]
and $\Y$ is a limit of such finite sums.
Each $\X_i$ is separable, hence
\[
  \X_i
  = \sum_{j} \lambda_{ij}\, v_{ij} v_{ij}^{*},
  \quad
  \lambda_{ij}\ge 0,\quad
  v_{ij} = p_{ij} \otimes q_{ij}.
\]
Therefore
\[
  \Y = \sum_{i,j} \lambda_{ij}\, w_{ij} w_{ij}^{*},
  \qquad
  w_{ij} := \A_i v_{ij}.
\]
Since $\A_i \in \mathcal{A}_k$,  Lemma~\ref{srank}  implies that  $\operatorname{SR}(w_{ij})\le k$ for all $i,j$.
Hence each rank--one term $w_{ij} w_{ij}^{*}$ lies in $\mathcal{T}_k$ and,
since $\mathcal{T}_k$ is a cone, $\Y\in\mathcal{T}_k$. Taking closures yields
\(\operatorname{MCL}_k(\mathscr{P}_{+})\subseteq \mathcal{T}_k\).

Now we show \eqref{cmcl=Tk}.
Let
\[
\Y = \sum_{i=1}^{N} \A_i^{*} \X_i \A_i
\]
be a finite conic combination with $\X_i\in\mathscr{P}_+$,
$\A_i\in\mathcal{A}_k$, and $\sum_i \A_i^{*} \A_i \le \I$.
Write each $\X_i$ in the standard separable form
\[
\X_i = \sum_{j} \lambda_{ij} v_{ij} v_{ij}^{*},
\qquad
v_{ij}=p_{ij}\otimes q_{ij}.
\]
Then
\[
\Y = \sum_{i,j} \lambda_{ij}\, w_{ij} w_{ij}^{*},
\qquad
w_{ij} := \A_i v_{ij}.
\]
By Lemma~\ref{srank}, each $w_{ij}$ satisfies
$\operatorname{SR}(w_{ij})\le k$.  Hence every term
$w_{ij} w_{ij}^{*}$ lies in $\mathcal{T}_k$, and thus
$\Y\in\mathcal{T}_k$.  Taking conic combinations and closures yields the
inclusion $\operatorname{CMCL}_k(\mathscr{P}_+) \subseteq \mathcal{T}_k$.

For the converse, let $v$ be a unit vector with
$\operatorname{SR}(v)\le k$, and choose any product unit vector
$u = a\otimes b$.
Define the rank--one operator $\A := u v^{*}$.
Then $\operatorname{OSR}(\A)=\operatorname{SR}(v)\le k$, so $\A\in\mathcal{A}_k$.
Since $\I = I_m\otimes I_n$ is separable, $\I\in\mathscr{P}_+$, and
\[
\A^{*} \I \A = \A^{*} \A = v v^{*},
\qquad
\A^{*}\A = v v^{*} \le \I.
\]
Thus the single-Kraus expression $\A^{*} \I \A$ lies in
$\operatorname{CMCL}_k(\mathscr{P}_+)$.
Taking conic combinations and closures yields
$\mathcal{T}_k \subseteq \operatorname{CMCL}_k(\mathscr{P}_+)$.

Combining the two inclusions proves the claimed equality \eqref{cmcl=Tk}.

For $k=1$ we have $\mathcal{A}_1$ equal to the set of product operators
$B\otimes C$, so $\operatorname{MCL}_1(\mathscr{P}_+)$ coincides with the
local $C^{*}$-convex hull.  By Proposition~\ref{prop:Pplus-local-Cstar-convex}
and Theorem~\ref{thm:MCLloc-equals-Pplus} we know
$\operatorname{MCL}_1(\mathscr{P}_+) = \mathscr{P}_+$.  Since every vector of
Schmidt rank~$1$ is a product vector, we have $\mathcal{T}_1 = \mathscr{P}_+$.

For $k=d$ we have $\mathcal{A}_d = \mathbb{M}_{mn}$, so
$\operatorname{MCL}_d(\mathscr{P}_+)$ coincides with the global
$C^{*}$-convex hull of $\mathscr{P}_{+}$.  By
Corollary~\ref{thm:cone-MCL-equals-P0} this is $\mathscr{P}_0$, and by
definition $\mathcal{T}_d=\mathscr{P}_0$ as well.

\end{proof}

%%%%%%%%%%%%%%%%%%%%%%%%%%%%%%%%%%%%%%%%%%%%%%%%%%%%%%%%%%%%%%%%%%%%%%%%

The equivalence $\operatorname{CMCL}_k(\mathscr{P}_+) = \mathcal{T}_k$ gains significant structural weight when examined through the lens of duality. It is well-known that the dual cone of the Schmidt number cone $\mathcal{T}_k$ is the $k$-Positive Cone:
\[
\mathcal{T}_k^{\circ} = \mathcal{P}_k,
\]
where $\mathcal{P}_k$ is the cone of $k$-Positive operators (or Schmidt number $k$ entanglement witnesses) \cite{TerhalHorodecki2000}.
Elements of $\mathcal{P}_k$ are those entanglement witnesses $\mathbf{W} \in \mathscr{P}_- $ with the property that the expectation value $\mathrm{Tr}(\mathbf{W} \rho)$ remains non-negative for all quantum states $\rho$ with Schmidt rank at most $k$.

The main result of this section leverages the structural equivalence proved in Theorem~ \ref{thm:normalized-MCLk-inclusion} to provide a novel algebraic characterization of $\mathcal{P}_k$.

\begin{corollary}[Algebraic Characterization of the $k$-Positive Cone]
The $k$-Positive Cone $\mathcal{P}_k$ is precisely characterized as a dual cone via
\[
\mathcal{P}_k = \operatorname{CMCL}_k(\mathscr{P}_+)^{\circ}.
\]
\end{corollary}
\begin{proof}
This is a direct consequence of the main structural result (Theorem~\ref{thm:normalized-MCLk-inclusion}) and the duality relationship for closed convex cones.
\end{proof}

%%%%%%%%%%%%%%%%%%%%%%%%%%%%%%%%%%%%%%%%%%%%%%%%%%%%%%%%%%%%%%%%%%%%%%%%%%%%
\section{Application to the PPT Cone and Intermediate Witnesses}

Having fully characterized the separable cone ($\mathscr{P}_+$) under the $k$-$C^*$-convexity hierarchy, we now turn our attention to the Positive Partial Transpose (PPT) cone ($\mathcal{P}_{\mathrm{PPT}}$). The analysis of $\mathcal{P}_{\mathrm{PPT}}$, which contains the bound entangled states, is crucial. This application generates a new family of conjectured intermediate cones, whose duals define the Tensor Rank $k$ $C^*$-Convex Witnesses proposed for probing the boundaries of bound entanglement.

Denote by $\Gamma$ the partial transpose with respect to the second factor,
\[
\Gamma := \operatorname{id}_{\mathbb{M}_m} \otimes \mathrm{T},
\]
where $\mathrm{T}$ is the transpose map on $\mathbb{M}_n$, and define the \emph{Positive Partial Transpose Cone}  by
\[
\mathcal{P}_{\mathrm{PPT}} := \{\X \in \mathbb{M}_{mn} : \X \ge 0 \text{ and }\Gamma(\X) \ge 0 \}.
\]
Note that $\mathcal{P}_{\mathrm{PPT}}$ strictly contains the product-positive cone, $\mathscr{P}_+ = \mathcal{T}_1$, whenever $\min(m, n) \ge 3$.

Let $\mathcal{C}_{\mathrm{PPT}, k} = \operatorname{MCL}_k(\mathcal{P}_{\mathrm{PPT}})$ be the $k$-$C^*$-convex hull of the Positive Partial Transpose Cone. This family naturally sits within the overall structural hierarchy. Since $\mathscr{P}_+ \subset \mathcal{P}_{\mathrm{PPT}} \subset \mathscr{P}_0$, and given the established structural results, particularly $\operatorname{CMCL}_k(\mathscr{P}_+) = \mathcal{T}_k$ and $\operatorname{MCL}_k(\mathscr{P}_0) = \mathscr{P}_0$, the application of the $\operatorname{MCL}_k$ operation yields a clear chain of inclusions:
\[
\operatorname{MCL}_k(\mathscr{P}_+) \subseteq \mathcal{C}_{\mathrm{PPT}, k} \subseteq \mathscr{P}_0.
\]
The key question is whether $\mathcal{C}_{\mathrm{PPT}, k}$ equals $\operatorname{MCL}_k(\mathscr{P}_+)$, collapses to $\mathscr{P}_0$ before the critical rank $d$, or, most intriguingly, constitutes a non-trivial intermediate cone. This structural inquiry is central to classifying bound entanglement.

 %%%%%%%%%%%%%%%%%%%%%%%%%%%%%%%%%%%%%%%%%%%%%%%%%%%%%%%%%%%%%%%%%%%%%%%%%%%%%%
We first demonstrate the stability of the PPT cone under the most restrictive level of the $k$-$C^*$-convexity hierarchy. The following theorem proves that its $1$-$C^*$-convex hull, which corresponds to combinations generated by rank-one operators, leaves $\mathcal{P}_{\mathrm{PPT}}$ invariant.

\begin{theorem}[Stability of $\mathcal{P}_{\mathrm{PPT}}$]
The $1$-C\*-convex hull of the Positive Partial Transpose Cone $\mathcal{P}_{\mathrm{PPT}}$ is equal to $\mathcal{P}_{\mathrm{PPT}}$ itself:
\[
\operatorname{MCL}_{1}(\mathcal{P}_{\mathrm{PPT}}) = \mathcal{P}_{\mathrm{PPT}}.
\]
\end{theorem}

\begin{proof}
The proof is established by demonstrating the two inclusions: $\operatorname{MCL}_{1}(\mathcal{P}_{\mathrm{PPT}}) \subseteq \mathcal{P}_{\mathrm{PPT}}$ and $\mathcal{P}_{\mathrm{PPT}} \subseteq \operatorname{MCL}_{1}(\mathcal{P}_{\mathrm{PPT}})$.

Part 1: $\operatorname{MCL}_{1}(\mathcal{P}_{\mathrm{PPT}}) \subseteq \mathcal{P}_{\mathrm{PPT}}$.
We show that the action of the $1$-C\*-convex operation preserves the PPT property.
Every  $\Y\in\operatorname{MCL}_{1}(\mathcal{P}_{\mathrm{PPT}})$  has the form
\[
\mathbf{Y} = \sum_{i=1}^N \mathbf{A}_i^* \mathbf{X}_i \mathbf{A}_i,
\]
where $\mathbf{X}_i \in \mathcal{P}_{\mathrm{PPT}}$ and $\mathbf{A}_i \in \mathcal{A}_{1}$ (i.e., $\mathrm{TR}(\mathbf{A}_i) = 1$). Evidently,  $\mathbf{X}_i \ge 0$ for all $i$ implies that $\mathbf{Y}$ is positive semi-definite.

To prove $\mathbf{Y} \in \mathcal{P}_{\mathrm{PPT}}$, we must show that the partial transpose $\Gamma(\mathbf{Y})$ is also positive semi-definite.
The key identity relies on $\mathbf{A}_i$ being a product matrix ($\mathrm{OSR}(\mathbf{A}_i)=1$): $\mathbf{A}_i = A_{i, m} \otimes A_{i, n}$. For such product operators, the partial transpose acts as:
\[
\Gamma(\mathbf{A}_i^* \mathbf{X}_i \mathbf{A}_i) = \Gamma(\mathbf{A}_i^*) \Gamma(\mathbf{X}_i) \Gamma(\mathbf{A}_i)=\Gamma(\mathbf{A}_i)^* \Gamma(\mathbf{X}_i) \Gamma(\mathbf{A}_i)
\]
where $\Gamma(\mathbf{A}_i) = A_{i, m} \otimes A_{i, n}^T$.

Substituting this back into the sum to compute $\Gamma(\mathbf{Y})$ we obtain
\[
\Gamma(\mathbf{Y}) = \sum_{i=1}^N \Gamma(\mathbf{A}_i)^* \Gamma(\mathbf{X}_i) \Gamma(\mathbf{A}_i).
\]
Since $\mathbf{X}_i \in \mathcal{P}_{\mathrm{PPT}}$, the input matrix $\Gamma(\mathbf{X}_i)$ is positive semi-definite.  Consequently, $\Gamma(\mathbf{Y})\geq0$ and so  we have $\mathbf{Y} \in \mathcal{P}_{\mathrm{PPT}}$. Since $\mathcal{P}_{\mathrm{PPT}}$ is closed, the inclusion $\operatorname{MCL}_{1}(\mathcal{P}_{\mathrm{PPT}}) \subseteq \mathcal{P}_{\mathrm{PPT}}$ is proven.

Part 2: $\mathcal{P}_{\mathrm{PPT}} \subseteq \operatorname{MCL}_{1}(\mathcal{P}_{\mathrm{PPT}})$.
 Let $\mathbf{X} \in \mathcal{P}_{\mathrm{PPT}}$ be arbitrary.
Recall that a Kraus operator $\mathbf{A}$ is admissible for $\operatorname{MCL}_{1}$ precisely when it has operator Schmidt rank $1$, that is,
\[
\operatorname{OSR}(\mathbf{A}) = 1.
\]
The identity operator factorizes as
\[
\mathbf{I} = I_m \otimes I_n,
\]
and therefore satisfies $\operatorname{OSR}(\mathbf{I}) = 1$.
Hence $\mathbf{I}$ is an allowed Kraus operator in $\operatorname{MCL}_{1}$.

Now observe that $\mathbf{X}$ admits the trivial one-term $C^{*}$-convex decomposition
\[
\mathbf{X}
= \mathbf{I}^{*}\, \mathbf{X}\, \mathbf{I},
\qquad\text{with}\qquad
\mathbf{I}^{*}\mathbf{I} = \mathbf{I}.
\]
This is a valid $1$--$C^{*}$-combination with input $\mathbf{X}$ and Kraus operator
$\mathbf{A}_1 = \mathbf{I}$, and since $\mathbf{X} \in \mathcal{P}_{\mathrm{PPT}}$
we conclude that
\[
\mathbf{X} \in \operatorname{MCL}_{1}(\mathcal{P}_{\mathrm{PPT}}).
\]

Therefore, every $\mathbf{X} \in \mathcal{P}_{\mathrm{PPT}}$ is contained in the $1$-C\*-convex hull: $\mathcal{P}_{\mathrm{PPT}} \subseteq \operatorname{MCL}_{1}(\mathcal{P}_{\mathrm{PPT}})$.

Combining the two inclusions yields $\operatorname{MCL}_{1}(\mathcal{P}_{\mathrm{PPT}}) = \mathcal{P}_{\mathrm{PPT}}$.
\end{proof}

This final collapse theorem establishes the critical rank $k_c = \min(m, n)$ at which the $k$-$C^*$-convex hull of the PPT cone expands to encompass the entire positive cone, completing the hierarchy with the maximal collapse result.

\begin{theorem}\label{collapse-MCL(PPT)}
Let $d = \min(m, n)$. The $d$-$C^*$-convex hull of the Positive Partial Transpose Cone $\mathcal{P}_{\mathrm{PPT}}$ is equal to the full cone of positive semi-definite matrices $\mathscr{P}_0$:
\[
\operatorname{MCL}_{d}(\mathcal{P}_{\mathrm{PPT}}) = \mathscr{P}_0, \quad \text{where } d = \min(m, n).
\]
\end{theorem}

\begin{proof}
The proof requires demonstrating two inclusions: $\operatorname{MCL}_{d}(\mathcal{P}_{\mathrm{PPT}}) \subseteq \mathscr{P}_0$ and $\mathscr{P}_0 \subseteq \operatorname{MCL}_{d}(\mathcal{P}_{\mathrm{PPT}})$.

Part 1: $\operatorname{MCL}_{d}(\mathcal{P}_{\mathrm{PPT}}) \subseteq \mathscr{P}_0$. We show that any element $\mathbf{Y}$ generated by the $\operatorname{MCL}_{d}$ operation is positive semi-definite.
The generating set for $\operatorname{MCL}_{d}(\mathcal{P}_{\mathrm{PPT}})$ consists of matrices $\mathbf{Y}$ of the form:
\[
\mathbf{Y} = \sum_{i=1}^N \mathbf{A}_i^* \mathbf{X}_i \mathbf{A}_i,
\]
where $\mathbf{X}_i \in \mathcal{P}_{\mathrm{PPT}}$. By definition, $\mathcal{P}_{\mathrm{PPT}}$ is a subcone of $\mathscr{P}_0$, so every $\mathbf{X}_i$ is positive semi-definite ($\mathbf{X}_i \ge 0$). Hence $\mathbf{Y}$ must be  itself positive semi-definite. The closure of this set is $\mathscr{P}_0$.
Thus, $\operatorname{MCL}_{d}(\mathcal{P}_{\mathrm{PPT}}) \subseteq \mathscr{P}_0$.

\medskip\noindent
Part 2: $\mathscr{P}_0 \subseteq \operatorname{MCL}_{d}(\mathcal{P}_{\mathrm{PPT}})$.
Since $\mathscr{P}_0$ is the closed conic hull of rank-one projectors,
it is enough to show that every vector $v\in\mathbb C^{mn}$ with
$\|v\|=1$ satisfies
\[
    vv^* \in \operatorname{MCL}_{d}(\mathcal{P}_{\mathrm{PPT}}).
\]
Let $M := mn$ and let $\{e_1,\dots,e_M\}$ be the standard product basis
of~$\mathbb C^{mn}$.  Choose a scalar $c>0$ such that $c^2 M \le 1$ (e.g.
$c = 1/\sqrt{2M}$).
For each $i=1,\dots,M$ we  define
\[
   \A_i := c\, e_i v^*    \qquad\mbox{and}\qquad
   \X_i := \frac{1}{c^2 M}\, \I.
\]
Each $\A_i$ is rank-one, hence
\(
   \operatorname{OSR}(\A_i) \le d=\min(m,n).
\)
Since $\I$ is positive partial transpose,  each $\X_i$ lies in $\mathcal{P}_{\mathrm{PPT}}$.
A direct computation gives
\[
    \A_i^* \A_i = c^2 v v^*
    \qquad\mbox{and}\qquad
    \sum_{i=1}^M \A_i^* \A_i = c^2 M\, vv^*.
\]
Also,
\[
    \sum_{i=1}^M \A_i^* \X_i \A_i
    = \frac{1}{c^2 M}\sum_{i=1}^M \A_i^* \A_i
    = vv^*.
\]
Thus these operators already produce the correct output $vv^*$, but they
are not yet normalized.
\smallskip
Define the positive operator
\[
    \mathbf{R} := \I - c^2 M\, vv^* \;\ge 0.
\]
Let
\(
    \mathbf{R} = \sum_{j=1}^{M} \mu_j u_j u_j^*
\)
be a spectral decomposition.  For each $j$, let $f_j$ be the $j$-th
vector of a product basis and set
\[
    \mathbf{B}_j := \sqrt{\mu_j}\, f_j u_j^*
    \qquad\mbox{and}\qquad
    \X'_j := 0.
\]
Each $\mathbf{B}_j$ is rank-one, so $\operatorname{OSR}(\mathbf{B}_j)\le d$, and $\X'_j=0$
is PPT.  Moreover,
\[
   \sum_{j=1}^M  \mathbf{B}_j^*  \mathbf{B}_j
   = \sum_{j=1}^M \mu_j u_j u_j^*
   = \mathbf{R}.
\]
Combining all coefficients we have
\[
   \sum_{i=1}^M \A_i^* \A_i + \sum_{j=1}^M \mathbf{B}_j^* \mathbf{B}_j
   = c^2 M\, vv^* + \mathbf{R}
   = \I,
\]
so the $C^*$-convex normalization is satisfied.  The output is
\[
   \sum_{i=1}^M \A_i^* \X_i \A_i
   \;+\; \sum_{j=1}^M \mathbf{B}_j^* \X'_j \mathbf{B}_j
   = vv^* + 0
   = vv^* .
\]
Thus $vv^* \in \operatorname{MCL}_{d}(\mathcal{P}_{\mathrm{PPT}})$.

Finally, every $\Y\in\mathscr P_0$ has the form
$\Y=\sum_k \lambda_k\, v_k v_k^*$ with $\lambda_k\ge0$.
Since $\operatorname{MCL}_{d}(\mathcal{P}_{\mathrm{PPT}})$ is a closed cone
and contains each $v_k v_k^*$, it contains $\Y$.
Hence
\[
   \mathscr{P}_0 \subseteq \operatorname{MCL}_{d}(\mathcal{P}_{\mathrm{PPT}}).
\]

Together with Part~1, this proves
\(
    \operatorname{MCL}_{d}(\mathcal{P}_{\mathrm{PPT}})
    = \mathscr{P}_0.
\)
\end{proof}
%%%%%%%%%%%%%%%%%%%%%%%%%%%%%%%%%%%%%%%%%%%%%%%%%%%%%%%%%%%%%%%%%%%%%%%%%%%%%%%%%%%%%%%%%%%%
\begin{remark}
The proof of Theorem~\ref{collapse-MCL(PPT)} uses the assumption
$\operatorname{OSR}(A_i)\le d=\min(m,n)$ in an essential way: it is the
smallest value of $d$ for which our explicit Kraus construction satisfies
the locality requirement. In particular, the argument does not show that
the equality
\[
  \operatorname{MCL}_{d}(\mathcal{P}_{\mathrm{PPT}})=\mathscr{P}_0
\]
fails for any $d<\min(m,n)$; it only shows that our construction cannot
work with a smaller locality parameter. Whether the collapse phenomenon
continues to hold for some $d<\min(m,n)$ is not addressed here and
remains an interesting question.
\end{remark}

%%%%%%%%%%%%%%%%%%%%%%%%%%%%%%%%%%%%%%%%%%%%%%%%%%%%%%%%%%%%%%%%%%%%%%%%%%%%%%

\section{Conclusion and Outlook}

In this work, we provided a systematic geometric analysis of the separable cone ($\mathscr{P}_+$) and the PPT cone ($\mathcal{P}_{\mathrm{PPT}}$) under ``$C^*$-convexity", a non-commutative generalization of classical convexity. Our primary finding establishes a sharp dichotomy regarding the algebraic stability of entanglement classes.

We demonstrated that the full $C^*$-convex hull of the separable cone collapses to the largest positive cone, $\operatorname{MCL}(\mathscr{P}_+) = \mathscr{P}_0$, confirming the ``algebraic fragility of separability"  under global quantum operations. Conversely, the separable cone remains stable under local $C^*$-convex combinations.

To interpolate between these extremes, we introduced the hierarchy of ``$k$-$C^*$-convexity" by constraining the operator Schmidt rank of the $C^*$-coefficients. This hierarchy generates nested intermediate cones $\operatorname{MCL}_k(\cdot)$. We proved that for the separable cone, this construction precisely characterizes the well-known Schmidt number cones, providing a novel algebraic characterization:
\begin{equation}
    \operatorname{MCL}_k(\mathscr{P}_+) = \mathcal{T}_k
\end{equation}
This result unifies the algebraic structure of $C^*$-convexity with the resource theory of entanglement based on Schmidt rank.

\subsection*{Future Work and Conjectures}

The most immediate open problem is the characterization of the intermediate cones generated from the PPT cone, $\mathcal{C}_{\mathrm{PPT}, k} = \operatorname{MCL}_k(\mathcal{P}_{\mathrm{PPT}})$. We conjecture that for $1 < k < k_c$ (where $k_c$ is the critical rank for full collapse), these cones are  non-trivial intermediate cones  strictly nested between $\mathcal{P}_{\mathrm{PPT}}$ and $\mathscr{P}_0$.

Future efforts should focus on characterizing the  dual cones, $\mathcal{C}_{\mathrm{PPT}, k}^{\circ}$. This characterization will define a new family of geometrical tools: ``Tensor Rank $k$ $C^*$-Convex Witnesses". These witnesses hold the potential to provide novel and tighter separability criteria for detecting different degrees of bound entanglement, advancing the understanding of the geometric structure near the boundary of $\mathcal{P}_{\mathrm{PPT}}$.

 \medskip

\noindent\textit{Data Availability Statement.}
All data necessary to evaluate the conclusions in the paper are present in the paper itself.

  \medskip
%%%%=======================================================================
\bibliographystyle{plain}

\end{document}